\documentclass[sn-mathphys,Numbered]{sn-jnl}
\usepackage{geometry}
\usepackage{leftidx}
\usepackage{graphicx}%
\usepackage{multirow}%
\usepackage[integrals]{wasysym}
\usepackage{amsmath,amssymb,amsfonts}%
\usepackage{amsthm}%
\newtheorem{theorem}{Theorem}
\newtheorem{proposition}{Proposition}%
\newtheorem{corollary}{Corollary}%
\newtheorem{remark}{Remark}%
\newtheorem{lemma}{Lemma}

  \newtheorem{definition}{Definition}%
\usepackage{mathtools}
\usepackage{mathrsfs}%
\usepackage[title]{appendix}%
\usepackage{xcolor}%
\usepackage{textcomp}%
\usepackage{manyfoot}%
\usepackage{booktabs}%
\usepackage{algorithm}%
\usepackage{algorithmicx}%
\usepackage{algpseudocode}%
\usepackage{listings}%
\usepackage{enumerate}
\usepackage{mathptmx}
\usepackage{mathrsfs}
\DeclareMathAlphabet{\mathcal}{OMS}{cmsy}{m}{n}
\usepackage{orcidlink}
\usepackage{hyperref}
\usepackage{ulem}
\hypersetup
    {
    colorlinks=true,
    linkcolor=blue,
    filecolor=red,      
    urlcolor=blue,
    linktoc=page,
    citecolor=magenta,
    }
    \usepackage{manyfoot}

\begin{document}

\title[Article Title]{Geodesic causality in Kerr spacetimes with $|a|\geq M$}

\author{Giulio Sanzeni$^{\dagger}$ \orcidlink{0009-0001-8226-8486}, Karim Mosani$^{\flat}$ \orcidlink{0000-0001-5682-1033}} 

\affil{$^{\dagger}$ Fakult\"at f\"ur Mathematik, Ruhr-Universit\"at Bochum,  Universit\"atsstra\ss e 150,  44801,  Bochum, Germany} 

\affil{$^\flat$ Fakultät für Mathematik, Universität Wien,  Oskar-Morgenstern-Platz 1, 1090 Wien, Austria}

\email{giulio.sanzeni@rub.de; karim.mosani@univie.ac.at}
\date{\today}

\abstract{
The analytic extension of the Kerr spacetimes into the negative radial region contains closed causal curves for any non-zero rotation parameter $a$ and mass parameter $M$. Furthermore, the spacetimes become totally vicious when $|a|>M$, meaning that through every point there exists a closed timelike curve. Despite this, we prove that the Kerr spacetimes do not admit any closed null geodesics when $|a|\geq M$. This result generalises recent findings by one of the authors, which showed the nonexistence of closed causal geodesics in the case $|a|<M$. Combining these results, we establish the absence of closed null geodesics in Kerr spacetimes for any non-zero $a$.
}

\keywords{Kerr spacetime, naked singularity, causality violation, closed null geodesics}
 
\maketitle
\tableofcontents

\section{Introduction}\label{sec1}

\noindent Kerr spacetimes \cite{Kerr-paper} are stationary, axisymmetric solutions of the vacuum Einstein's field equations. They are described in terms of two physical parameters: $M$,  the \textit{mass parameter},  and $a$,  the \textit{rotation parameter} (angular momentum per unit mass). In 1968, Carter \cite{Carter_causality} showed that the analytic extension of the Kerr solution in the negative radial region exhibits causality violations for all non-zero values of $a$: both timelike and null closed curves exist. However, one of the authors recently proved that in the slow-rotating ($|a|<M$) case, the analytic extension called \textit{Kerr-star spacetime} contains neither closed null geodesics \cite{sanzeni_null} nor closed timelike geodesics \cite{sanzeni_timelike}. The existence of the horizons in the slow-rotating case plays a crucial role in proving these results. The horizons are null hypersurfaces; hence, future light cones are always bent towards the centre of the black hole. Therefore, future-directed causal curves can cross the horizons at most once, so any causal curve intersecting a horizon cannot be closed (see Proposition $5.3$ of \cite{sanzeni_null}). This argument, for instance, is used to rule out the closeness of null geodesics with vanishing energy $E$, see \S $5.4$ of \cite{sanzeni_null}. Another essential element in the proofs is the existence of a causal region (see Corollary $2.17$ of \cite{sanzeni_null}), which precludes closed causal geodesics strictly contained within this subset, see Lemma $5.1$ of \cite{sanzeni_null}.

The problem of the existence of closed null geodesics in the analytic extension through the negative radial region of the Kerr spacetimes with $|a|\geq M$  remained open. When $|a|=M$, the Kerr spacetime describes a black hole with a single horizon, and it is called \textit{extreme Kerr-star spacetime}. Hence, the extreme Kerr-star spacetime differs from the slow-rotating one, which admits two distinct horizons. However, as in the slow-rotating case, there exists a causal region. Thus, it turns out that a minor modification of the proof given in \cite{sanzeni_null} is sufficient to prove the nonexistence of closed null (lightlike) geodesics in the extreme Kerr-star spacetime.

When $|a|>M$, the Kerr spacetime describes a naked singularity (see, for instance, \S 9.2 in \cite{hawking_ellis}, and also p.144 of \cite{Joshi_93(2)} for the definition of naked singularity in terms of causal boundary points as introduced by Geroch, Kronheimer and Penrose \cite{Geroch_1972}), as no horizon exists and is referred to as the \textit{Kerr-naked singularity spacetime} or \textit{fast-rotating Kerr spacetime}. Such a spacetime is known to be totally vicious (see, for instance, \cite{Clarke_1982}); namely, there exists a closed timelike curve through every point of the spacetime. Instead, in the slow-rotating case, the only vicious region is the subset $\{r<r_{-}\}$, where $r$ is the usual Boyer--Lindquist radial coordinate and $\{r=r_{-}\}$ is the Cauchy horizon of the black hole. Due to the absence of horizons and the total viciousness of the Kerr-naked singularity spacetime, the proofs from \cite{sanzeni_null,sanzeni_timelike} do not apply in their current form to the fast-rotating case. However, following the methodology of \cite{sanzeni_null}, we demonstrate the nonexistence of closed null geodesics in the Kerr-naked singularity spacetime.

As far as naked singularities are concerned, they could possibly form as an end state of gravitational collapse \cite{Penrose_naked}. In 1969, Penrose
\cite{Perose_1969}
proposed what is now known as the Cosmic Censorship Conjecture, which forbids their existence. However, there are examples of ``fairly reasonable" initial conditions that lead to their formation
\cite{Christodoulou_1984, Joshi_1993, Mosani_2020}.
A crucial aspect about naked singularities is their stability, namely, whether a small perturbation to such initial conditions still entails a naked singularity scenario. For instance, the results in \cite{Christodoulou_1999} suggest that under the hypothesis of spherical symmetry and in the presence of a massless scalar field, such singularities are not stable under arbitrarily small perturbations to the initial conditions. Nevertheless, these and similar constructions provide valuable motivation
for further investigation into the nature and formation of naked singularities in general relativity. 

Returning to the Kerr-naked singularity spacetime, in \S$66$, Ch. $6$ \cite{Chandrasekhar} Chandrasekhar posed the following question: Does such a spacetime allow causality violation along causal geodesics? Furthermore, he argued that $r$-unbounded geodesics with negative Carter constant and negative $r$-turning points may violate a particular notion of causality, as discussed in \cite{calvani_defelice}. Since no horizons exist, such geodesics could depart from a positive asymptotic $r$-region, reach a negative $r$-turning point in the region where $\partial_\phi$ is timelike, and return to the asymptotic $r$-region at an earlier $t$-coordinate (see Eq. ($384$) of \cite{Chandrasekhar}). Partial results in this direction were obtained in \cite{defelice_bini}, where the authors showed that null fly-by geodesics with negative Carter constant lying on $\theta=\textrm{const.}$ hyperplanes cannot decrease their $t$-coordinate since they reach $r$-turning points before being able to enter the region in which their $t$-coordinate could reverse. 
As a result, the authors suggest that vortical 
(i.e. with negative Carter constant) null boreal (i.e. with $\theta=\textrm{const.}$) geodesics may satisfy (some) ``causality principle". 
Observe, however, that the authors in \cite{deFelice_Calvani_violation} claimed that some particular vortical null geodesics may violate this notion of causality. Note that in our work, we adopt the standard definition of ``causality”, namely the absence of closed causal curves. Therefore, no geodesic with unbounded $r$-behaviour can violate causality by producing closed geodesics.

Finally, we observe a similarity between the Kerr-naked singularity spacetime and the G\"{o}del spacetime \cite{Godel}: both are totally vicious, yet neither of them contains closed null geodesics, Refs. \cite{Chandr_Wright,Kundt,Nolan_godel}.

\vspace{0.3cm}

\subsection*{Result}\label{subsection Result}
Given a spacetime $\big(\mathcal{M},\mathbf{g}\big)$,  \textit{i.e.} a time-oriented connected Lorentzian manifold,  and a geodesic curve $\gamma:I=[a,b]\rightarrow \mathcal{M}$,  we say that $\gamma$ is a \textit{closed geodesic}  if $\gamma(a)=\gamma(b)$ and $\gamma'(a)=\lambda\gamma'(b)\neq 0$,  for some real number $\lambda\neq 0$. The purpose of this paper is to prove the nonexistence of closed null (lightlike) geodesics in the analytic extension through the negative radial region of the Kerr spacetimes with $|a|\geq M$.
\newline

\begin{theorem}\label{main theorem}
There are no closed null (lightlike) geodesics in both the Kerr-naked singularity ($|a|>M$) spacetime and the extreme  Kerr-star ($|a|=M$) spacetime.
\end{theorem}
\vspace{0.3cm}

The organisation of the paper is as follows. In  \S \ref{definiton of Kerr},  we introduce the Kerr metric and discuss the properties of the Kerr spacetimes.  In  \S \ref{study of geodesic equations}, we recall the equations of motion satisfied by the geodesics and study the properties of null geodesics in the Kerr-naked singularity spacetime required to prove the main theorem.  In \S \ref{section: main theorem},  we give a proof of Theorem \ref{main theorem} in the case $|a|>M$. In \S\ref{extreme case section}, we give a proof of Theorem \ref{main theorem} in the case $|a|=M$.

\section{\textbf{The Kerr-naked singularity spacetime}}\label{definiton of Kerr}

Consider $\mathbb{R}^2\times S^2$ with coordinates $(t,r)\in\mathbb{R}^2$ and $(\theta,\phi)\in S^2$.  Fix two real numbers $a\in\mathbb{R}\setminus \{0\}$,  $M\in \mathbb{R}_{>0}$  and define the functions 
\[
 \rho(r,\theta):= \sqrt{r^2+a^2\cos^2\theta}
\]
and 
\[
\Delta(r):=r^2-2Mr+a^2.
\]
\vspace{0.2cm}

\noindent We focus first on the case $|a|>M$ corresponding to the \textit{Kerr-naked singularity spacetime} (also known as \textit{fast-rotating Kerr spacetime}),  for which $\Delta(r)$ has no real roots,  in particular $\Delta(r)>0$ for all $r$.  Define the \textit{ring singularity} as

\[ \Sigma:=\{\rho(r,\theta)=0\}=\{r=0,\;\theta=\pi/2\}.\]

\noindent The {\it Kerr metric}  \cite{Kerr-paper} in {\it Boyer--Lindquist coordinates} is

\begin{align}\label{kerr metric}
\mathbf{g}:=-&\bigg(1-\frac{2 M r}{\rho^2(r,\theta)} \bigg)\:dt\otimes dt-\frac{4Mar\sin^2\theta}{\rho^2(r,\theta)}\: dt\otimes d\phi + \nonumber\\+ &\bigg(r^2+a^2+\frac{2Mra^2\sin^2\theta}{\rho^2(r,\theta)} \bigg)\sin^2\theta\:  d\phi\otimes d\phi + \frac{\rho^2(r,\theta)}{\Delta(r)}\: dr\otimes dr + \rho^2(r,\theta)\: d\theta\otimes d\theta.
\end{align}

Observe that if $|a|>M$ the vector field $V:=(r^2+a^2)\partial_t + a\partial_\phi$ is timelike throughout the manifold by Lemma $2.1.3$ of \cite{KBH_book}.  We define then a time-orientation on $(\mathbb{R}^2\times S^2,\mathbf{g})$ by declaring the vector field $V$  to be future-directed. 

\vspace{0.3cm}
\begin{definition}[\cite{sanzeni_timelike}]\label{definition kerr spacetime}
A {\it Kerr spacetime} is an analytic spacetime $(\mathcal{M}_{Kerr},g)$ such that 
\begin{enumerate}
\item[(1)]  there exists a family of open disjoint isometric embeddings $\Phi_i \colon \mathcal{B}_i\hookrightarrow \mathcal{M}_{Kerr}$ $(i\in \{1,2,3\})$ of Boyer--Lindquist (BL) blocks $(\mathcal{B}_i,
\mathbf{g}|_{\mathcal{B}_i})$, such that $\cup_{i
} \Phi_i(\mathcal{B}_i)$ is dense in $\mathcal{M}_{Kerr}$;

\item[(2)] there are analytic functions $r$ and $C$ on $\mathcal{M}_{Kerr}$  such that their restriction on each $\Phi_i(\mathcal{B}_i)$ of condition $(1)$ is $\Phi_i$-related to the BL functions $r$ and $C=\cos\theta$ on $\mathcal{B}_i$;

\item[(3)]  there is an isometry $\epsilon: \mathcal{M}_{Kerr}\rightarrow \mathcal{M}_{Kerr}$ called the \textit{equatorial isometry} whose restrictions to each BL block sends $\theta$ to $\pi-\theta$,  leaving the  other coordinates unchanged;

\item[(4)] there are Killing vector fields $\tilde{\partial}_t$ and $\tilde{\partial}_\phi$ on $\mathcal{M}_{Kerr}$ that restrict  to the BL coordinate vector fields $\partial_t$ and $\partial_\phi$ on each BL block.
\end{enumerate}
\end{definition}

\vspace{0.3cm}
\begin{definition}\label{abstract def axis and Eq}
In a Kerr spacetime 
\begin{enumerate}
\item the \textit{axis} $A=\{\theta=0,\pi\}$ is the set of zeroes of the Killing vector field $\tilde{\partial}_\phi$ as in $(4)$ of Definition \ref{definition kerr spacetime};
\item the \textit{equatorial hyperplane} $Eq=\{\theta=\pi/2\}$ is the set of fixed points of the equatorial isometry $\epsilon$ as in $(3)$ of Definition \ref{definition kerr spacetime}.
\end{enumerate}
\end{definition}
\vspace{0.3cm}
\begin{lemma}[\cite{KBH_book}, Theorem $1.7.12$]\label{A and Eq closed totally geod subman}
The axis $A$ and the equatorial hyperplane $Eq$ of a Kerr spacetime as in Definition \ref{definition kerr spacetime} are closed totally geodesic submanifolds of $\mathcal{M}_{Kerr}$.
\end{lemma}

\vspace{0.3cm}
\begin{definition}\label{Kerr-naked spacetime}
The \textit{Kerr-naked singularity spacetime} is a Kerr spacetime as defined in \ref{definition kerr spacetime} given by the tuple $(K_N,\mathbf{g},o)$ with $K_N=\{(t,r)\in\mathbb{R}^2,\, (\theta,\phi)\in S^2\}\setminus\Sigma$,  $\mathbf{g}$ as in Eq.  \eqref{kerr metric} (extended over the axis) with $|a|>M$ and $o$ is the future time-orientation induced by the timelike vector field $V$.
\end{definition}
\vspace{0.3cm}
\begin{remark}
In the Kerr-naked singularity spacetime, the function $\Delta(r)$ appearing in the metric tensor \eqref{kerr metric} is strictly positive for all $r$.  Therefore, we can describe the whole spacetime just in terms of the Boyer--Lindquist coordinates functions, unlike the slow-rotating and extreme cases, which require the introduction of a new set of coordinates that cover the horizons, the Kerr-star coordinates (see Definition $2.9$ in \cite{sanzeni_null}). Hence, the Kerr-naked singularity spacetime consists of a single BL block.
\end{remark}

\vspace{0.3cm}

The analytic extension through the negative $r$-region of the Kerr spacetime for any $a\neq 0$ produces causality violations, as one can easily see from $\partial_\phi$ becoming timelike for negative radii sufficiently close to $r=0$, see Eq. \eqref{kerr metric}. In the Kerr-star spacetime $K^*$ with $|a|<M$, it is known that given any pair of points in the BL block $\{r<r_{-}\}$ there exists a future-directed timelike curve from one point to the other (see Proposition $2.4.7$, \cite{KBH_book}), hence through every point of $\{r<r_{-}\}$ there exists a closed timelike curve. The proof of this statement requires that the closed timelike curve in $\{r<r_{-}\}$ enters the \textit{Carter time machine} region $\mathfrak{T}:=\{p\in K^*:\mathbf{g}_p(\partial_\phi,\partial_\phi)<0\}\subset\{r<0\}$, see Lemma $2.4.8$, \cite{KBH_book}. In this Lemma one first shows the existence of future-directed timelike curve from any $p=(r_0,\theta_0,\phi_0,t_0)\in\{r<r_{-}\}=: III$ to $(r^*,\pi/2,\phi,t_0+\Delta t)\in\mathfrak{T}$, and then the existence of a past-directed timelike curve from the same $p$ to $(r^*,\pi/2,\phi,t_0-\Delta t)\in\mathfrak{T}$, for some positive $\Delta t$, and some negative $r^*$. However, in the Kerr-naked singularity case $|a|>M$, no horizon exists, so one can freely choose the point $p$ everywhere in $\mathbb{R}^2\times S^2\setminus\Sigma$ and hence find the required past-directed timelike curve. Therefore, we have the following statement.
\vspace{0.3cm}
\begin{proposition}
    The Kerr-naked singularity spacetime, as in Definition \ref{Kerr-naked spacetime}, is totally vicious, i.e. there exists a closed timelike curve through every point.
\end{proposition}

\section{\textbf{Geodesics in Kerr-naked singularity spacetime}}\label{study of geodesic equations}


Recall that the Kerr-naked singularity spacetime $(K_N,\mathbf{g})$ admits two Killing vector fields $\partial_t$ and $\partial_\phi$.
\vspace{0.3cm}
\begin{definition}[\textit{Energy and angular momentum}] 
For a geodesic $\gamma$ of $(K_N,\mathbf{g})$,  the constants of motion 

\[
E=E(\gamma):=-\mathbf{g}(\gamma',\partial_t) \quad
\text{and} \quad
L=L(\gamma):=\mathbf{g}(\gamma',\partial_\phi)
\]
are called its {\it energy} and its {\it angular momentum (around the axis of rotation)}, respectively.
\end{definition}

\vspace{0.3cm}

\begin{definition}
The \textit{canonical vector fields} are
\[
V:=(r^2+a^2)\partial_t + a\partial_\phi\quad \text{and }\quad W:=\partial_\phi + a \sin^2\theta\,\partial_t.
\]
\end{definition}

\begin{definition}\label{definitions of P and D}
Let $\gamma$ be a geodesic with energy $E$ and angular momentum $L$.  Define the functions $\mathbb{P}$ and $\mathbb{D}$ along $\gamma$ by
\[
\mathbb{P}(r):=-\mathbf{g}(\gamma',V)=(r^2+a^2)E-La
\]
and 
\[
\mathbb{D}(\theta):=\mathbf{g}(\gamma',W)=L-Ea\sin^2\theta.
\]
\end{definition}

\vspace{0.3cm}

\noindent A geodesic in a Kerr spacetime has two additional constants of motions: the {\it Lorentzian energy} $q:=\mathbf{g}(\gamma',\gamma')$ and $K$ \cite{Carter_causality}, which can be defined (see Ch.  $7$ in \cite{Chandrasekhar}) by

\[
K:=2\rho^2(r,\theta)\mathbf{g}(l,\gamma')\mathbf{g}(n,\gamma')+r^2q,
\]
where $l=\frac{1}{\Delta(r)}V+\partial_r$ and $n=\frac{1}{2\rho^2(r,\theta)}V-\frac{\Delta(r)}{2\rho^2(r,\theta)}\partial_r$.

\vspace{0.3cm}

\begin{definition}[\textit{Carter constant}]
On a Kerr spacetime,  the constant of motion
\[
Q:=K-(L-aE)^2\hspace{0.5cm}\textrm{or}\hspace{0.5cm} \mathcal{Q}:=Q/E^2\hspace{0.3cm} \textrm{if}\hspace{0.3cm}E\neq 0
\]
is called the Carter constant.  
\end{definition}

\subsection*{Equations of motion}

\begin{proposition}[\cite{KBH_book},  Proposition $4.1.5$,  Theorem $4.2.2$] \label{differential equations of geodessics}
Let $\gamma$ be a geodesic with constants of motion $E,  L,  Q,  q$.   Then the components of $\gamma$ in the BL coordinates $(t,r,\theta,\phi)$ satisfy the following set of \textit{first} order differential equations

\begin{align}
\begin{cases}
 \rho^2(r,\theta)\phi'=\frac{\mathbb{D}(\theta)}{\sin^2\theta}+a\frac{\mathbb{P}(r)}{\Delta(r)} \\ \rho^2(r,\theta) t'= a\mathbb{D}(\theta) + (r^2+a^2)\frac{\mathbb{P}(r)}{\Delta(r)} \label{geodes diff equations}\\ \rho^4(r,\theta) r'^2 = R(r)\\  \rho^4(r,\theta) \theta'^2 = \Theta (\theta)   
 \end{cases}
\end{align}
where

\begin{align*}
R(r):=&\Delta(r)\left[(qr^2-K(E,L,Q)\right]+\mathbb{P}^2(r) \\ 
=& (E^2+q)r^4 -2Mqr^3 + \mathfrak{X}(E,L,Q) r^2 + 2MK(E,L,Q)r - a^2 Q,\label{other form of R(r)}\\
\Theta(\theta):=& K(E,L,Q)+qa^2 \cos^2\theta -\frac{ \mathbb{D}(\theta)^2}{\sin^2\theta} \\
=& Q + \cos^2\theta \left[ a^2(E^2+q)-L^2/\sin^2\theta\right],
\end{align*}
with
\[
\mathfrak{X}(E,L,Q):=a^2(E^2+q)-L^2-Q\text{\quad and\quad} K(E,L,Q)=Q+(L-aE)^2.
\]

\end{proposition}
\vspace{0.3cm}
\begin{remark}
Notice that $\Theta(\theta)$ is also well-defined if the null geodesic crosses $A=\{\theta=0,\pi\}$.  Indeed,  $L=0$ since $\partial_\phi\equiv 0$ on $A$. Hence $\mathbb{D}(\theta)=-Ea\sin^2\theta$. So
  
\[ 
\Theta(\theta)=K(E,0,Q)-(-Ea\sin^2\theta)^2/\sin^2\theta=Q+a^2E^2-a^2E^2\sin^2\theta=Q+a^2E^2\cos^2\theta.\\
\]

\end{remark} 

Since in the third and in the fourth differential equations of Proposition \ref{differential equations of geodessics} the left-hand sides are clearly non-negative,  we observe that the functions $R(r)$ and $\Theta(\theta)$ are non-negative along the geodesics.  Hence, the geodesic motion can only happen in the $r,\theta$-region for which $R(r),\Theta(\theta)\geq 0$. The non-negativity of $R(r)$ and $\Theta(\theta)$ can be used to study the dynamics of the $r,\theta$-coordinates of the geodesics,  together with the next proposition.
\vspace{0.3cm}

\begin{proposition}[\cite{KBH_book},  Corollary $4.3.8$]\label{initial conditions and zeroes}
Suppose $R(r_0)=0$.  Let $\gamma$ be a geodesic whose $r$-coordinate satisfies the initial conditions $r(\gamma(s_0))=r_0$ and $(r\circ\gamma)'(s_0)=0$.
\begin{enumerate}
\item If $r_0$ is a multiplicity one zero of $R(r)$,  \textit{i.e.} $R'(r_0)\neq 0$,  then $r_0$ is an $r$-turning point,  namely $(r\circ\gamma)'(s)$ changes sign at $s_0$.
\item If $r_0$ is a higher order zero of $R(r)$,  \textit{i.e.} at least $R'(r_0)= 0$,  then $\gamma$ has constant $(r\circ\gamma)(s)=r_0$.  
\end{enumerate}
Analogous results hold for $r$ and $R(r)$ replaced by $\theta$ and $\Theta(\theta)$.  
\end{proposition}

\subsection*{\textbf{Properties of null geodesics in Kerr-naked singularity spacetime}}\label{section: properites of null geodesics}

Since the vector fields $V, W, \partial_r,\partial_\theta$ are mutually orthogonal,  the tangent vector to a geodesic $\gamma$ can be decomposed as $\gamma'=\gamma'_\Pi + \gamma'_\perp$ where $\Pi:= span \{ \partial_r,  V \}$ (timelike plane) and $\Pi^\perp:= span \{ \partial_\theta,  W \}$ (spacelike plane).
\vspace{0.3cm}
\begin{definition}
A Kerr geodesic $\gamma$ is said to be \textit{principal} if $\gamma' = \gamma'_\Pi$. 
\end{definition}
\vspace{0.3cm}
\begin{proposition}[\cite{KBH_book}, Lemma $4.2.7$]\label{expression for K}
    Let $K$ and $q$ be constants of motion associated to a geodesic $\gamma$. Then $K=\rho^2(r,\theta)\mathbf{g}(\gamma'_{\perp},\gamma'_{\perp})-qa^2\cos^2\theta.$
\end{proposition}
\vspace{0.3cm}
\begin{corollary}[\cite{KBH_book},  Corollary $4.2.8(2)$] \label{eq principal nulls}
If $\gamma$ is a null geodesic,  then $K\geq 0$,  and $K=0$ if and only if $\gamma$ is principal.
\end{corollary}
\vspace{0.3cm}
\begin{remark}
In the slow-rotating $|a|<M$ Kerr spacetime,  the function $\Delta(r)$ in the metric tensor \eqref{kerr metric} can vanish at some positive roots $r=r_{\pm}$. At these radii, there exist null hypersurfaces called horizons $\mathscr{H}_{\pm}=\{r=r_{\pm}\}$. Null geodesics lying in the horizons are called \textit{restphotons},  and they are integral curves of  $V|_{\mathscr{H}_{\pm}}$ by Proposition $2.5.5$ of \cite{KBH_book}. Indeed, when $|a|<M$, the vector field $V$ becomes null on the horizons.
In the Kerr-naked singularity spacetime, however, there are no horizons,  hence no restphotons.
\end{remark}
\vspace{0.3cm}
\begin{proposition}\label{E=L=K=0 proposition}
Let $\gamma$ be a null geodesic in the Kerr-naked singularity spacetime.  If $\textrm{Im}~\gamma\subset A$,  then $K=L=0$. Moreover,  there are no null geodesics with $K=E=0.$
\end{proposition}

\begin{proof}
    On $A$, $\partial_\phi\equiv 0$, hence $L=\mathbf{g}(\gamma',\partial_\phi)=0$. Since $\textrm{Im}~\gamma\subset A$, $\gamma'_{\perp}=0$. Hence $K=0$ by Proposition \ref{expression for K} because $\gamma$ is null.\\
    Suppose by contradiction that there exists a null geodesic with $K=E=0$. Since $\gamma$ is principal by Corollary \ref{eq principal nulls}, $\gamma'=c_r \partial_r + c_V V$, for some $c_r,c_V\in\mathbb{R}$. Then $0=E=-\mathbf{g}(\gamma',\partial_t)=-c_V\mathbf{g}(V,\partial_t)=c_V \Delta(r)$ (see p. $60$ of \cite{KBH_book}). Since for Kerr-naked singularity spacetime $\Delta(r)>0$ for all $r$, we must have $c_V=0$, hence $\gamma'=c_r\partial_r$. However, $\partial_r$ is always spacelike in the Kerr-naked singularity spacetime.
\end{proof}
\vspace{0.3cm}
\begin{remark}
    In the slow-rotating and extreme Kerr-star spacetimes, null geodesics with $K=E=0$ exist. For instance, when their energy $E$ vanishes, they correspond to the restphotons, see Lemma $4.2.9$ of \cite{KBH_book}.
\end{remark}
\vspace{0.3cm}

The following propositions about null geodesics are valid in any Kerr spacetime, regardless of the magnitude of the rotation parameter $a$.
\vspace{0.3cm}

\begin{proposition}[\cite{sanzeni_null}, Proposition $4.6$]\label{Q<0 condition}
Let $\gamma$ be a null geodesic with $Q<0$.  Then

\begin{enumerate}
\item $\gamma$ does not intersect $Eq=\{\theta=\pi/2\}$;
\item $a^2E^2> L^2$ and in particular $E\neq 0$.
\end{enumerate}
\end{proposition}
\vspace{0.3cm}
\begin{proposition}[\cite{sanzeni_null}, Proposition $4.8$]\label{geod Q<0}
For $Q<0$ null geodesics,  $R(r)$ is convex and has either no roots or two (maybe coinciding) negative roots. 
\end{proposition}

\section{\textbf{Proof of Theorem  \ref{main theorem}}: $|a|>M$ case} \label{section: main theorem}

By contradiction, let $\gamma\colon I\to K_N$ be a closed null geodesic (CNG).
Since the radius function $r\colon K_N\to\mathbb{R}$ is everywhere smooth, the composition $r\circ \gamma$ must have at least two critical points $s_0<s_1$ in each period $[a,a+T)$, i.e.  
$(r\circ \gamma)'(s_0)=(r\circ \gamma)'(s_1)=0$. Since $\rho\colon K_N\to\mathbb{R}$  never vanishes, the differential equation for $r\circ \gamma$
\[
(\rho\circ \gamma)^4[(r\circ \gamma)']^2=R(r\circ \gamma)
\]
implies that $R(r\circ \gamma(s_{0,1}))=0$.  Because of the differential equation,  the geodesic motion must happen in the $r$-region on which $R(r\circ \gamma)\geq 0$.  Further since $R$ is a polynomial in $r$ we can distinguish two cases:
\begin{enumerate}

\item The zeros $r\circ \gamma(s_{0,1})$ of $R$ are simple, i.e. $dR/dr\neq 0$ at these points. Then $r\circ \gamma(s_{0,1})$ are turning points of $r\circ \gamma$, i.e. 
$(r\circ \gamma)'$ changes its sign at $s_0$ and $s_1$.  

\item One of the zeros $r\circ \gamma(s_{0})$ or $r\circ \gamma(s_{1})$ is a higher order zero of $R$. Then $r\circ \gamma$ is constant.
\end{enumerate}
Both facts follow from Proposition \ref{initial conditions and zeroes}.  

\vspace{0.3cm}
\begin{proposition}\label{proposition spacelike foliation}
There are no closed null geodesics strictly contained in the region $\{r\geq 0\}$.
\end{proposition}

\begin{proof}
First we claim that the hypersurfaces $\mathcal{N}_t:=\{t=\textrm{const}\}\cap\{r\geq 0\}\subset K_N$ are spacelike.  Indeed,  if $p\in\mathcal{N}_t\setminus A$,  where $A=\{\theta=0,\pi\}$,  then $T_p\mathcal{N}_t$ is spanned by $\partial_r,\partial_\theta,\partial_\phi$ which are spacelike and orthogonal to each other.  If $p\in A\cap\mathcal{N}_t$,  then $p=(t,r,q)$ with $q=(0,0,\pm 1)\in S^2\subset \mathbb{R}^3$,  and we may replace $\partial_\theta,\partial_\phi$ by any basis of $T_qS^2$.  Suppose by contradiction that there exists a CNG $\gamma$ in $\{r\geq 0\}$.  Then there would exist $t_0,s_0$ such that $\gamma'(s_0)\in T_{\gamma(s_0)}\mathcal{N}_{t_0}$.  This is a contradiction since $\gamma'(s_0)$ is null.
\end{proof}

\begin{remark}
    Proposition \ref{proposition spacelike foliation} specifically holds for the Kerr-naked singularity spacetime, since the coordinate vector field $\partial_r$ is everywhere spacelike if $\Delta(r)>0$ for all $r$ (compare with Proposition $5.2$ in \cite{sanzeni_null}).
\end{remark}

\hspace{1cm}

\subsection*{Null geodesics in the axis $A$}

First, we rule out CNGs entirely contained in the axis $A=\{\theta=0,\pi\}$.
\vspace{0.3cm}
\begin{proposition} \label{no closed geodesics in axis}
There are no CNGs which are tangent at some point to $A=\{\theta=0,\pi\}$.  In particular,  there are no CNGs entirely contained in $A$.
\end{proposition}

\begin{proof}
First of all,  $A=\{\theta=0,\pi\}$ is a $2$-dimensional closed totally geodesic submanifold by Lemma \ref{A and Eq closed totally geod subman}. Hence, if a geodesic $\gamma$ is tangent to $A$ at some point,   it will always lie on $A$.  By Proposition \ref{E=L=K=0 proposition},  if $\textrm{Im}~\gamma\subset A$,  then $L=K=0$.  Hence, the geodesic must have $E\neq 0$, since there are no null geodesics with $K=E=0$ by Proposition \ref{E=L=K=0 proposition}.  Then using Eq.  \eqref{geodes diff equations},  we have

\begin{align*}
R(r)=E^2(r^2+a^2)^2>0.
\end{align*}
So $R(r)$ has no zeroes, and therefore the geodesic cannot be closed.

\end{proof}

\subsection*{Steps of the proof for other cases}\label{steps of other cases}

The proof splits into two main cases  $E=0$ and $E\neq 0$. When $E\neq 0$,  we analyse the three subcases $Q=0$,  $Q>0$ and $Q<0$.\\

\subsection*{Case $E=0$}\label{section $E=0$}

\vspace{0.3cm}

From Proposition \ref{differential equations of geodessics},  for null $(q=0)$ geodesics we have

\begin{align}
\label{R equation E=0}R(r)&=\mathfrak{X}(0,L,Q)r^2+2MK(0,L,Q)r-a^2Q\geq 0,\\
\label{theta equation E=0}\Theta(\theta)&=Q-\frac{\cos^2\theta}{\sin^2\theta}L^2\geq 0,
\end{align}
with $\mathfrak{X}(0,L,Q)=-(L^2+Q)$ and $K(0,L,Q)=L^2+Q$,  so  $\mathfrak{X}=-K$.  From Eq. \eqref{theta equation E=0}, we must have $Q\geq 0$,  hence $K(0,L,Q)=L^2+Q\geq 0$,  as already known from Corollary \ref{eq principal nulls}.  By Proposition \ref{E=L=K=0 proposition}, we can assume $K(0,L,Q)>0$. We may also assume that the discriminant of \eqref{R equation E=0} is 

\[\textrm{dis}=4M^2(L^2+Q)^2-4a^2Q(L^2+Q)\geq 0,
\] 
as otherwise the geodesic cannot have bounded $r$-behaviour. Therefore, the two (maybe coinciding) roots of $R(r)$ are

\begin{align*}
M\pm\sqrt{M^2-\frac{a^2Q}{L^2+Q}}\geq 0,
\end{align*}
since $Q/(L^2+Q)\geq 0$.  Hence $\gamma$ is either bouncing between two non-negative simple roots or has a fixed $r$-coordinate at positive radius $r=M$. Therefore, it cannot be closed by Proposition \ref{proposition spacelike foliation}.

\subsection*{Case $E\neq 0$}\label{section $E div 0$}
\vspace{0.3cm}
\subsubsection*{Subcase $Q=0$} \label{case Q=0}

We have

\begin{align}
R(r)&=E^2r^4+(a^2E^2-L^2)r^2+2M(L-aE)^2r\geq 0,\\
\Theta(\theta)&=\cos^2\theta \bigg(a^2E^2-\frac{L^2}{\sin^2\theta}\bigg)\geq 0. \label{theta eq E non 0, Q=0}
\end{align}

\begin{proposition}[\cite{sanzeni_null}, Proposition $5.6$]\label{proposition r-bounded eq geodesics}
All the null geodesics with $E\neq 0,\; Q=0$ which have bounded $r$-behaviour lie in $Eq=\{\theta=\pi/2\}$.
\end{proposition}

\vspace{0.3cm}

By Proposition \ref{proposition r-bounded eq geodesics},  we can suppose $\theta(s)=\pi/2$ for every $s$.  Since $R'(0)=2M(L-aE)^2$,   if $L\neq aE$  geodesics with initial position in $\{r<0\}$ are constrained in this region and cannot reach $\{r\geq 0\}$.  Hence, there are no $r$-bounded null geodesics in the region $r<0$ by  Lemma $4.14.2$ \cite{KBH_book}.  Since the geodesics are constrained in $Eq$,  $r=0$ cannot be an $r$-turning point.  So if geodesics have bounded $r$-behaviour in $\{r>0\}$,  they cannot be closed by Proposition \ref{proposition spacelike foliation}.  If instead $L=aE$,  then $R(r)=E^2r^4$.  Then the only possible bounded $r$-behaviour would be $r(s)\equiv 0$, which cannot happen as otherwise the geodesic would lie on the ring singularity $\Sigma=\{r=0,\theta=\pi/2\}$.

\vspace{0.3cm}

\subsubsection*{Subcase $Q>0$} \label{case Q>0} From Corollary \ref{eq principal nulls},  we know that for null geodesics $K=Q+(L-aE)^2\geq 0$.  Hence if $Q>0$,  then $K>0.$ Let us again consider $R(r)=E^2r^4+\mathfrak{X}(E,L,Q)r^2+2MK(E,L,Q)r-a^2Q$.\\
\noindent Since $R(0)=-a^2Q<0$,  if a null geodesic has bounded $r$-behaviour  either it must be strictly contained in $\{r<0\}$ or strictly contained in $\{r>0\}$,  since we must have $R(r)\geq 0$.  Bounded $r$-behaviour in $\{r<0\}$ is impossible in this subcase.  Indeed,  the signs of the coefficients of $R(-r)$ are $+\; \textrm{sign}(\mathfrak{X})\; -\; -$,  so for every sign of $\mathfrak{X}$,  there would be only one change of sign,  hence there can be only one single real negative root by the \textit{``Descartes' rule of signs"}.  If instead the geodesics have bounded $r$-behaviour in the region $\{r>0\}$,  they cannot be closed by Proposition \ref{proposition spacelike foliation}.

\vspace{0.3cm}
\subsubsection*{Subcase $Q<0$} \label{case Q<0}

To rule out the closeness of such geodesics, we can follow the proof in \S $5.5.3$, \cite{sanzeni_null} which does not rely on the magnitude of the parameter $a$. We summarise here the main arguments. First, one considers the only possible bounded $r$-behaviour of null geodesics with $Q<0$, namely spherical geodesics at negative radii, i.e. $(r\circ\gamma)(s)=\textrm{const}<0$ for all $s$ (see Proposition \ref{geod Q<0}). Using the constancy of the $r$-coordinate, the $\theta$-differential equation in Proposition \ref{differential equations of geodessics} implies that the $\theta$-motion must be either constant or periodic. Closeness of null geodesics with $(\theta\circ\gamma)(s)=\textrm{const}$ for all $s$ is excluded in Proposition $5.7$ of \cite{sanzeni_null}. The possible periodic $\theta$-motions are discussed in Proposition $5.9$ (see Figures $11,12,13,14$ of \cite{sanzeni_null}). Then we parametrise the angular momentum and the Carter constant of the spherical null geodesics with negative Carter constant by the fixed negative $r$-coordinate of the geodesic, see Proposition $5.10$ of \cite{sanzeni_null}. Finally, we compute the variation of the $t$-coordinate over a full $\theta$-oscillation, as just a function of the fixed $r$-coordinate, see Eq. ($29$) of \cite{sanzeni_null}. Using relations between elliptic integrals and hypergeometric functions, we prove that such $t$-variation is strictly positive for every constant negative radius (see discussion following Eq. ($29$) of \cite{sanzeni_null}).

\vspace{0.3cm}

\section{Proof of Theorem \ref{main theorem}: $|a|=M$ case}\label{extreme case section}

The extreme Kerr-star spacetime is a Kerr-star spacetime as in Definition $2.12$ of \cite{sanzeni_null} with $|a|=M$. Under this assumption the function $\Delta(r)$ in the metric tensor \eqref{kerr metric} (also shown in Remark $2.1$ of \cite{sanzeni_null}) vanishes only at $r=r_{\mathscr{H}}:=M$. Therefore we can think of it as the Kerr-star spacetime with only two Boyer--Lindquist (BL) blocks: $\{r<r_{\mathscr{H}}\}$, corresponding to the vicious BL block $III$, and $\{r>r_{\mathscr{H}}\}$, corresponding to the causal BL block $I$. The null hypersurface $\{r=r_{\mathscr{H}}\}$ is hence the only horizon of the spacetime.
Since the metric tensor \eqref{kerr metric} expressed in BL coordinates is not defined on the horizon, one has to introduce a new set of coordinates which cover the latter, the Kerr-star coordinates (see Definition $2.9$ of \cite{sanzeni_null}). Similar to the slow-rotating ($|a|<M$) case, the extreme Kerr-star ($|a|=M$) spacetime features both a horizon and a causal region. Therefore, the proof of the nonexistence of closed null geodesics presented in \S $5$ of \cite{sanzeni_null} can also be applied to the extreme Kerr-star spacetime, with appropriate modifications.

\vspace{0.3cm}

To prove Theorem \ref{main theorem} in the case \(|a| = M\), one must modify certain arguments presented in \cite{sanzeni_null}, as outlined below. First, instead of relying on Corollary $2.17$ from \cite{sanzeni_null}, one can utilize the causality of the region $\{r\geq r_{\mathscr{H}}\}$. Additionally, in Lemma $5.1$, Proposition $5.2$ and Proposition $5.3$ of \cite{sanzeni_null}, the Cauchy horizon $r=r_{-}$ should be replaced with the only horizon present in the extreme case, namely $r=r_{\mathscr{H}}$. It is important to note that Proposition $4.5$ of \cite{sanzeni_null} remains valid in the extreme case as well (see also Proposition \ref{E=L=K=0 proposition} for comparison). Therefore, with the above changes the proof strategy laid out in \S $5$ of \cite{sanzeni_null} can be followed, with a single necessary modification: the treatment of the case $E=0$ with $K(0,L,Q)=Q+L^2>0$ (refer to \S$5.4$ of \cite{sanzeni_null}). In this case, one must have $Q\geq 0$. For this class of geodesics, a bounded $ r$-behaviour is only possible if $L\neq 0$. Otherwise, if $L=0, ~Q>0$, the associated polynomial becomes $R(r)=-Q\Delta(r)$ and the only motion corresponding to a bounded $r$-coordinate is a spherical one at $r=r_{\mathscr{H}}=M$, namely a restphoton. However by Proposition $4.5$ of \cite{sanzeni_null}, restphotons are uniquely identified by $E=L=K(E,L,Q)=0$.

Indeed, if $L\neq 0$ and $Q\geq 0$, the discriminant of $R(r)$ is $\textrm{dis}=4M^2(L^2+Q)^2-4a^2Q(L^2+Q)>0$ since $|a|=M\neq 0,\,L^2+Q>Q\geq 0$. Hence if $E=0,\,K(0,L,Q)>0,\,L\neq 0$, the geodesics describe $r$-bouncing orbits between two simple roots of the associated polynomial $R(r)$: one root lies below $r=r_{\mathscr{H}}=M$, and the other lies above it. However, Proposition $5.3$ of \cite{sanzeni_null} can be invoked to exclude the closeness of such geodesics.

\vspace{0.3cm}

\section{Conclusions}\label{sec5}

We have proved the nonexistence of closed null geodesics in both the fast ($|a|>M$) and extreme ($|a|=M$) Kerr spacetimes, analytically extended in the negative radial region. The proof follows analogous strategies to those used in the slow-rotating Kerr spacetime \cite{sanzeni_null}, suitably adapted for the other cases. Combining Theorem \ref{main theorem} with the result in \cite{sanzeni_null}, we can state that for any rotation parameter $a$, the Kerr spacetimes do not admit closed null geodesics.

The existence of closed timelike geodesics in both fast-rotating and extreme cases remains an open problem. For instance, 
the nonexistence proof for timelike geodesics given in \cite{sanzeni_timelike} cannot be applied to the fast-rotating case. To exclude the possibility of closed spherical timelike geodesics with negative Carter constant in the slow-rotating case, a bound on the $r$-coordinate of admissible closed geodesics is employed (see Proposition $5.7$ of \cite{sanzeni_timelike}): closed geodesics may have only existed at radii $-M<r<0$. In the proof by contradiction it is therefore assumed
that $-M<r<0$ in establishing the positivity of the variation of the $t$-coordinate over a full $\theta$-oscillation for both classes of spherical timelike geodesics considered in Proposition $5.12$ of \cite{sanzeni_timelike} (see also Lemma $5.13$ and Proposition $5.17$ in \cite{sanzeni_timelike}, which assume $-M<r<0$). However, the same $r$-bound from Proposition $5.7$ of \cite{sanzeni_timelike} cannot be guaranteed when $|a|>M$, since the coordinate vector field $\partial_\phi$ becomes spacelike when $r\leq -|a|<-M$ (see Lemma $2.4.9$ of \cite{KBH_book}).  Therefore, closed timelike geodesics can only be ruled out if $r\leq -|a|$.
As a result, it remains unclear how to rule out the possibility of closed spherical timelike geodesics with negative Carter constant at fixed negative $r$-coordinate in the range $-|a|<r<-M$.

\backmatter
\bmhead*{Acknowledgement} We thank Stefan Suhr for reading and commenting on the manuscript.
\bmhead*{Fundings} This research was funded in part by the Deutsche Forschungsgemeinschaft (DFG, German Research Foundation)  [Project-ID 281071066 – TRR 191] and by the Austrian Science Fund (FWF) [Grant DOI \href{https://www.fwf.ac.at/en/research-radar/10.55776/EFP6}{10.55776/EFP6}]. For open access purposes, the authors have applied a CC BY public copyright license to any author-accepted manuscript version arising from this submission.

\bmhead*{Data availability} Data sharing is not applicable to this article as no datasets were generated or analysed during the current study.

\bmhead*{Conflict of interest} The authors have no competing interests to declare that are relevant to the content of this article.

\bibliography{main}
\end{document}